\documentclass[11pt]{article}%
\usepackage{amsfonts}
\usepackage{amsmath}
\usepackage{amssymb}
\usepackage{graphicx}
\usepackage[colorlinks=true,linkcolor=blue,citecolor=red,plainpages=false,pdfpagelabels]%
{hyperref}%
\providecommand{\U}[1]{\protect\rule{.1in}{.1in}}
\newtheorem{theorem}{Theorem}

\newtheorem{corollary}[theorem]{Corollary}

\newtheorem{proposition}[theorem]{Proposition}

\newenvironment{proof}[1][Proof]{\noindent\textbf{#1.} }{\ \rule{0.5em}{0.5em}}
\begin{document}

\title{Optimal uniform continuity bound for conditional entropy of classical--quantum states}
\author{Mark M.\ Wilde\thanks{Hearne Institute for Theoretical Physics, Department of
Physics and Astronomy, Center for Computation and Technology, Louisiana State
University, Baton Rouge, Louisiana 70803, USA. Email: mwilde@lsu.edu}}
\maketitle

\begin{abstract}
In this short note, I show how a recent result of Alhejji and Smith
[arXiv:1909.00787] regarding an optimal uniform continuity bound for classical
conditional entropy leads to an optimal uniform continuity bound for quantum
conditional entropy of classical--quantum states. The bound is optimal in the
sense that there always exists a pair of classical--quantum states saturating
the bound, and so no further improvements are possible. An immediate
application is a uniform continuity bound for entanglement of formation that
improves upon the one previously given by Winter in [arXiv:1507.07775]. Two
intriguing open questions are raised regarding other possible uniform
continuity bounds for conditional entropy, one about quantum--classical states
and another about fully quantum bipartite states.

\end{abstract}

Recently, the following bound was established by Alhejji and Smith in
\cite{AS19} for $\varepsilon\in(0,1-1/\left\vert \mathcal{Y}\right\vert ]$:%
\begin{equation}
\left\vert H(Y|X)_{p}-H(Y|X)_{q}\right\vert \leq\varepsilon\log_{2}(\left\vert
\mathcal{Y}\right\vert -1)+h_{2}(\varepsilon),\label{eq:al-sm}%
\end{equation}
where $h_{2}(\varepsilon):=-\varepsilon\log_{2}\varepsilon-\left(
1-\varepsilon\right)  \log_{2}(1-\varepsilon)$ is the binary entropy, $p_{XY}$
and $q_{XY}$ are joint probability distributions over the finite-cardinality
alphabets $\mathcal{X}$ and $\mathcal{Y}$,
\begin{equation}
H(Y|X)_{p}:=-\sum_{x\in\mathcal{X}}p_{X}(x)\sum_{y\in\mathcal{Y}}%
p_{Y|X}(x)\log_{2}p_{Y|X}(y|x)
\end{equation}
and $H(Y|X)_{q}$ (defined in a similar way but with $q_{XY}$) are conditional
Shannon entropies, and
\begin{equation}
\varepsilon\geq\frac{1}{2}\Vert p_{XY}-q_{XY}\Vert_{1}:=\frac{1}{2}\sum
_{x\in\mathcal{X},y\in\mathcal{Y}}\left\vert p_{XY}(x,y)-q_{XY}%
(x,y)\right\vert .
\end{equation}
The quantity on the right-hand side is known as the total variational distance of
the probability distributions $p_{XY}$ and $q_{XY}$, and it is a measure of
their statistical distinguishability. The bound in \eqref{eq:al-sm} is called
a uniform continuity bound because the right-hand side depends only on
$\varepsilon$ and the cardinality~$\left\vert \mathcal{Y}\right\vert $. It is
optimal in the sense that for every $\varepsilon$ and $\left\vert
\mathcal{Y}\right\vert $, there exists a pair of distributions $p_{XY}$ and
$q_{XY}$ saturating the upper bound (see Eqs.~(27)--(28) of \cite{AS19}). It
generalizes the optimal uniform continuity bound for unconditional Shannon
entropy established independently by \cite[Eq.~(4)]{Z07}\ and~\cite{A07}.

Uniform continuity bounds of the form in \eqref{eq:al-sm} for both the
classical and quantum cases find application in providing estimates for
various communication capacities of classical and quantum channels
\cite{LS08,SSWR15,PhysRevLett.120.160503,LKDW18,KW17,SWAT18,KSW19,KGW19}.
Motivated by this application (as well as fundamental concerns), there has
been a large amount of work on this topic over the years
\cite{Fannes73,AF04,A07,Winter15,Shi15,Sh18,Shirokov_2018,Sh19}.

In this brief note, I show how to employ the bound in \eqref{eq:al-sm} to
establish the following optimal uniform continuity bound for conditional
entropy of finite-dimensional classical--quantum states, improving (optimally)
upon one of the cases given in Lemma~2 of \cite{Winter15}:

\begin{proposition}
\label{prop:main}The following inequality holds for $\varepsilon
\in(0,1-1/d_{B}]$:%
\begin{equation}
\left\vert H(B|X)_{\rho}-H(B|X)_{\sigma}\right\vert \leq\varepsilon\log
_{2}(d_{B}-1)+h_{2}(\varepsilon), \label{eq:cont-cq}%
\end{equation}
where $d_{B}$ is the dimension of system $B$, the states $\rho_{XB}$ and
$\sigma_{XB}$ are the following finite-dimensional classical--quantum states:%
\begin{equation}
\sum_{x\in\mathcal{X}}r(x)|x\rangle\langle x|_{X}\otimes\rho_{B}^{x}%
,\qquad\sum_{x\in\mathcal{X}}s(x)|x\rangle\langle x|_{X}\otimes\sigma_{B}^{x},
\label{eq:cq1}%
\end{equation}
$r(x)$ and $s(x)$ are probability distributions, $\left\{  \rho_{B}%
^{x}\right\}  _{x}$ and $\left\{  \sigma_{B}^{x}\right\}  _{x}$ are sets of
states, the conditional entropy is defined in terms of the von Neumann entropy
as $H(B|X)_{\rho}:=\sum_{x}r(x)H(\rho_{B}^{x})$, and%
\begin{equation}
\varepsilon\geq\frac{1}{2}\left\Vert \rho_{XB}-\sigma_{XB}\right\Vert _{1}.
\end{equation}
Also, there exists a pair of classical--quantum states saturating the bound
for every value of $d_B$ and $\varepsilon\in(0,1-1/d_{B}]$.
\end{proposition}

\begin{proof}
The desired inequality is reduced to the classical case by means of a
conditional dephasing channel and data processing. This generalizes an
approach recalled in the introduction of \cite{Winter15}, which is attributed
therein to \cite{Petz08}. Suppose without loss of generality that
$H(B|X)_{\rho}\leq H(B|X)_{\sigma}$. Let a spectral decomposition of $\rho
_{B}^{x}$ be as follows:%
\begin{equation}
\rho_{B}^{x}=\sum_{y}r(y|x)|\phi^{y,x}\rangle\langle\phi^{y,x}|_{B},
\end{equation}
where $r(y|x)$ is a conditional probability distribution and $\left\{
|\phi^{y,x}\rangle_{B}\right\}  _{y}$ is a set of orthonormal states (for
fixed $x$). Define the conditional dephasing channel as%
\begin{equation}
\overline{\Delta}_{XB}^{\text{cd}}(\omega_{XB})=\sum_{x,y}\left(
|x\rangle\langle x|_{X}\otimes|\phi^{y,x}\rangle\langle\phi^{y,x}|_{B}\right)
\omega_{XB}\left(  |x\rangle\langle x|_{X}\otimes|\phi^{y,x}\rangle\langle
\phi^{y,x}|_{B}\right)  ,
\end{equation}
which we think of intuitively as dephasing or measuring system $X$ and then
based on the outcome, dephasing system $B$ in the eigenbasis of $\rho_{B}^{x}%
$. This is a unital channel, and so the entropy of any state on systems $X$
and $B$ does not decrease under its action. When this conditional dephasing
acts on $\sigma_{XB}$, it leads to the following state:%
\begin{equation}
\overline{\Delta}_{XB}^{\text{cd}}(\sigma_{XB})=\sum_{x\in\mathcal{X}%
,y\in\mathcal{Y}}s(x)s(y|x)|x\rangle\langle x|_{X}\otimes|\phi^{y,x}%
\rangle\langle\phi^{y,x}|_{B},
\end{equation}
where $s(y|x)$ is a conditional probability distribution and $\mathcal{Y}$ is
an alphabet with the same cardinality as the dimension $d_{B}$: $\left\vert
\mathcal{Y}\right\vert =d_{B}$. Observe that%
\begin{equation}
\sigma_{X}=\operatorname{Tr}_{B}[\sigma_{XB}]=\operatorname{Tr}_{B}%
[\overline{\Delta}_{XB}^{\text{cd}}(\sigma_{XB})].\label{sx-inv}%
\end{equation}
Furthermore, the state $\rho_{XB}$ is invariant under the action of the
conditional dephasing channel:%
\begin{equation}
\rho_{XB}=\overline{\Delta}_{XB}^{\text{cd}}(\rho_{XB}).
\end{equation}
Observe that $\rho_{XB}$ and $\overline{\Delta}_{XB}^{\text{cd}}(\sigma_{XB})$
are commuting states, and thus can be considered as classical--classical
states (to be more precise, the first is classical and the second is classical
conditioned on the classical value in the first system). Define the joint
distributions $r_{XY}(x,y)=r(x)r(y|x)$ and $s_{XY}(x,y)=s(x)s(y|x)$. From
\eqref{sx-inv}\ and the fact that the conditional dephasing channel is unital,
it follows that%
\begin{align}
H(B|X)_{\sigma} &  =H(BX)_{\sigma}-H(X)_{\sigma}\\
&  =H(BX)_{\sigma}-H(X)_{\overline{\Delta}^{cd}(\sigma)}\\
&  \leq H(BX)_{\overline{\Delta}^{\text{cd}}(\sigma)}-H(X)_{\overline{\Delta
}^{\text{cd}}(\sigma)}\\
&  =H(B|X)_{\overline{\Delta}^{\text{cd}}(\sigma)}\\
&  =H(Y|X)_{s}.
\end{align}
So we have that%
\begin{equation}
H(Y|X)_{r}=H(B|X)_{\rho}\leq H(B|X)_{\sigma}\leq H(Y|X)_{s},
\end{equation}
which means that%
\begin{equation}
H(B|X)_{\sigma}-H(B|X)_{\rho}\leq H(Y|X)_{s}-H(Y|X)_{r}.
\end{equation}
Meanwhile, we have from data processing for normalized trace distance that%
\begin{align}
\frac{1}{2}\left\Vert \rho_{XB}-\sigma_{XB}\right\Vert _{1} &  \geq\frac{1}%
{2}\left\Vert \overline{\Delta}_{XB}^{\text{cd}}(\rho_{XB})-\overline{\Delta
}_{XB}^{\text{cd}}(\sigma_{XB})\right\Vert _{1}\\
&  =\frac{1}{2}\left\Vert \rho_{XB}-\overline{\Delta}_{XB}^{\text{cd}}%
(\sigma_{XB})\right\Vert _{1}\\
&  =\frac{1}{2}\left\Vert r_{XY}-s_{XY}\right\Vert _{1}.
\end{align}
In turn, this means that the following bound holds for total variational
distance:
\begin{equation}
\frac{1}{2}\left\Vert r_{XY}-s_{XY}\right\Vert _{1}\leq\varepsilon.
\end{equation}
Now we have completed the reduction to the classical case and invoke
\eqref{eq:al-sm} to conclude that%
\begin{align}
\left\vert H(B|X)_{\rho}-H(B|X)_{\sigma}\right\vert  &  =H(B|X)_{\sigma
}-H(B|X)_{\rho}\\
&  \leq H(Y|X)_{s}-H(Y|X)_{r}\\
&  \leq\varepsilon\log_{2}(d_{B}-1)+h_{2}(\varepsilon),
\end{align}
completing the proof of \eqref{eq:cont-cq}. The inequality in
\eqref{eq:cont-cq} is seen to be tight by using the classical example from
Eqs.~(27)--(28) of \cite{AS19}.
\end{proof}

\bigskip

By employing the same method of proof given for Corollary~4 in \cite{Winter15} (and observing that $\delta = \sqrt
{\varepsilon\left(  2-\varepsilon\right)}$ and $\delta \in (0,1-1/d]$ imply that $\varepsilon\in(0,1-\frac{\sqrt{2d-1}}{d}]$), we arrive at the following uniform continuity bound for entanglement of formation:

\begin{corollary}
Let $\rho_{AB}$ and $\sigma_{AB}$ be finite-dimensional quantum states such
that%
\begin{equation}
\frac{1}{2}\left\Vert \rho_{AB}-\sigma_{AB}\right\Vert _{1}\leq\varepsilon,
\end{equation}
where $\varepsilon\in(0,1-\frac{\sqrt{2d-1}}{d}]$ and $d=\min\left\{  d_{A},d_{B}\right\}  $.
Then%
\begin{equation}
\left\vert E_{F}(\rho_{AB})-E_{F}(\sigma_{AB})\right\vert \leq\delta\log
_{2}(d-1)+h_{2}(\delta),
\end{equation}
where $E_{F}$ is the entanglement of formation and $\delta=\sqrt
{\varepsilon\left(  2-\varepsilon\right)  }$. The entanglement of formation of
a state $\omega_{AB}$ is defined as follows \cite{BDSW96}:
\begin{multline}
E_{F}(\omega_{AB}):=\\
\inf\{H(B|X)_{\tau}:\tau_{XAB}=\sum_{x}p(x)|x\rangle\langle x|_{X}\otimes
\phi_{AB}^{x},\operatorname{Tr}_{X}[\tau_{XAB}]=\omega_{AB}\}.
\end{multline}
where each $\phi_{AB}^{x}$ is a pure state and $p(x)$ is a probability distribution.
\end{corollary}

The statement in Proposition~\ref{prop:main}\ has a straightforward
generalization to the case in which the classical conditioning system is
countable (thus addressing an open question stated in \cite{AS19}). To arrive
at the corollary, let us define conditional entropy in this case as follows:%
\begin{equation}
H(B|X)_{\rho}:=\sum_{x\in\mathcal{X}}p_{X}(x)H(\rho_{B}^{x}),
\label{eq:ce-main-formula}%
\end{equation}
where $\rho_{XB}$ has the same form as in \eqref{eq:cq1}, except that
$\mathcal{X}$ is now a countable alphabet (correspondingly, $X$ is now a
separable Hilbert space). Then we have the following corollary:

\begin{corollary}
The following inequality holds for $\varepsilon\in(0,1-1/d_{B}]$:%
\begin{equation}
\left\vert H(B|X)_{\rho}-H(B|X)_{\sigma}\right\vert \leq\varepsilon\log
_{2}(d_{B}-1)+h_{2}(\varepsilon),
\end{equation}
where $d_{B}$ is the dimension of system $B$, the states $\rho_{XB}$ and
$\sigma_{XB}$ are the following classical--quantum states:%
\begin{equation}
\sum_{x\in\mathcal{X}}r(x)|x\rangle\langle x|_{X}\otimes\rho_{B}^{x}%
,\qquad\sum_{x\in\mathcal{X}}s(x)|x\rangle\langle x|_{X}\otimes\sigma_{B}^{x},
\label{eq:cq-state-cor}%
\end{equation}
with system $B$ finite-dimensional and the alphabet $\mathcal{X}$ countable,
$r(x)$ and $s(x)$ are probability distributions, $\left\{  \rho_{B}%
^{x}\right\}  _{x}$ and $\left\{  \sigma_{B}^{x}\right\}  _{x}$ are sets of
states, and%
\begin{equation}
\varepsilon\geq\frac{1}{2}\left\Vert \rho_{XB}-\sigma_{XB}\right\Vert _{1}.
\label{eq:TD-assumption}%
\end{equation}

\end{corollary}

\begin{proof}
Recall that the conditional entropy of a bipartite state $\rho_{LM}$\ acting
on a separable Hilbert space, with $H(L)_{\rho}<\infty$, is defined as
\cite{K11}%
\begin{equation}
H(L|M)_{\rho}:=H(L)_{\rho}-I(L;M)_{\rho},\label{eq:ce-in-terms-of-ent-mi}%
\end{equation}
where the mutual information is given in terms of the relative entropy
$D(\omega\Vert\tau)$ \cite{F70,Lindblad1973}\ of states\ $\omega$ and $\tau$
as%
\begin{align}
I(L;M)_{\rho} &  :=D(\rho_{LM}\Vert\rho_{L}\otimes\rho_{M}),\label{eq:MI}\\
D(\omega\Vert\tau) &  :=\frac{1}{\ln2}\sum_{x,y}\left\vert \langle\phi
_{x}|\psi_{y}\rangle\right\vert ^{2}\left[  \lambda_{x}\ln(\lambda_{x}/\mu
_{y})+\mu_{y}-\lambda_{x}\right]  ,\label{eq:rel-ent}%
\end{align}
and spectral decompositions of states $\omega$ and $\tau$ are given by%
\begin{equation}
\omega=\sum_{x}\lambda_{x}|\phi_{x}\rangle\langle\phi_{x}|,\qquad\tau=\sum
_{y}\mu_{y}|\psi_{y}\rangle\langle\psi_{y}|.
\end{equation}
Let us first verify that the formula in
\eqref{eq:ce-in-terms-of-ent-mi}\ reduces to that in
\eqref{eq:ce-main-formula}. Evaluating the formulas in \eqref{eq:MI} and
\eqref{eq:rel-ent} for the case of interest (the state $\rho_{XB}$ in
\eqref{eq:cq-state-cor}), while taking spectral decompositions of $\rho_{XB}$
and $\rho_{X}\otimes\rho_{B}$ as%
\begin{align}
\rho_{XB} &  =\sum_{x\in\mathcal{X}}r(x)|x\rangle\langle x|_{X}\otimes
\sum_{y\in\mathcal{Y}}r(y|x)|\phi^{y,x}\rangle\langle\phi^{y,x}|_{B},\\
\rho_{X}\otimes\rho_{B} &  =\sum_{x^{\prime}\in\mathcal{X}}r(x^{\prime
})|x^{\prime}\rangle\langle x^{\prime}|_{X}\otimes\sum_{z\in\mathcal{Z}%
}q(z)|\psi_{z}\rangle\langle\psi_{z}|_{B},
\end{align}
with $\mathcal{X}$ countable, $\mathcal{Y}$ and $\mathcal{Z}$ finite, we find
that%
\begin{align}
I(X;B)_{\rho} &  =\frac{1}{\ln2}\sum_{x,y,z,x^{\prime}}\left\vert \left(
\langle x^{\prime}|_{X}\otimes\langle\psi_{z}|_{B}\right)  \left(
|x\rangle_{X}\otimes|\phi^{y,x}\rangle_{B}\right)  \right\vert ^{2}\nonumber\\
&  \qquad\times\left[  r(x)r(y|x)\ln\left(  \frac{r(x)r(y|x)}{\left[
r(x^{\prime})q(z)\right]  }\right)  +r(x^{\prime})q(z)-r(x)r(y|x)\right]  \\
&  =\frac{1}{\ln2}\sum_{x,y,z}\left\vert \langle\psi_{z}|\phi^{y,x}\rangle
_{B}\right\vert ^{2}\nonumber\\
&  \qquad\times\left[  r(x)r(y|x)\ln\left(  \frac{r(x)r(y|x)}{\left[
r(x)q(z)\right]  }\right)  +r(x)q(z)-r(x)r(y|x)\right]  \\
&  =\frac{1}{\ln2}\sum_{x}r(x)\sum_{y,z}\left\vert \langle\psi_{z}|\phi
^{y,x}\rangle_{B}\right\vert ^{2}\left[  r(y|x)\ln\left(  \frac{r(y|x)}%
{q(z)}\right)  +q(z)-r(y|x)\right]
\end{align}
For every $x\in\mathcal{X}$, we find that%
\begin{align}
&  \sum_{y,z}\left\vert \langle\psi_{z}|\phi^{y,x}\rangle_{B}\right\vert
^{2}\left[  r(y|x)\ln\left(  \frac{r(y|x)}{q(z)}\right)  +q(z)-r(y|x)\right]
\\
&  =\sum_{y,z}\left\vert \langle\psi_{z}|\phi^{y,x}\rangle_{B}\right\vert
^{2}\left[  r(y|x)\ln\left(  \frac{r(y|x)}{q(z)}\right)  \right]  \\
&  =\sum_{y,z}\left\vert \langle\psi_{z}|\phi^{y,x}\rangle_{B}\right\vert
^{2}\left[  r(y|x)\ln\left(  r(y|x)\right)  \right]  +\sum_{y,z}\left\vert
\langle\psi_{z}|\phi^{y,x}\rangle_{B}\right\vert ^{2}\left[  r(y|x)\ln
\!\left(  \frac{1}{q(z)}\right)  \right]  \\
&  =\sum_{y}\left[  r(y|x)\ln\left(  r(y|x)\right)  \right]  +\sum_{z}%
\langle\psi_{z}|\rho_{B}^{x}|\psi_{z}\rangle\ln\!\left(  \frac{1}%
{q(z)}\right)  \\
&  =-\left(  \ln2\right)  H(\rho_{B}^{x})+\sum_{z}\langle\psi_{z}|\rho_{B}%
^{x}|\psi_{z}\rangle\ln\!\left(  \frac{1}{q(z)}\right)  .
\end{align}
Then we find that%
\begin{align}
I(X;B)_{\rho} &  =\sum_{x\in\mathcal{X}}r(x)\left[  -H(\rho_{B}^{x})+\sum
_{z}\langle\psi_{z}|\rho_{B}^{x}|\psi_{z}\rangle\log_{2}\!\left(  \frac
{1}{q(z)}\right)  \right]  \\
&  =-\sum_{x\in\mathcal{X}}r(x)H(\rho_{B}^{x})+\sum_{z}\langle\psi_{z}|\left[
\sum_{x}r(x)\rho_{B}^{x}\right]  |\psi_{z}\rangle\log_{2}\!\left(  \frac
{1}{q(z)}\right)  \\
&  =-\sum_{x\in\mathcal{X}}r(x)H(\rho_{B}^{x})+\sum_{z}\langle\psi_{z}%
|\rho_{B}|\psi_{z}\rangle\log_{2}\!\left(  \frac{1}{q(z)}\right)  \\
&  =-\sum_{x\in\mathcal{X}}r(x)H(\rho_{B}^{x})+\sum_{z}q(z)\log_{2}\!\left(
\frac{1}{q(z)}\right)  \\
&  =-\sum_{x\in\mathcal{X}}r(x)H(\rho_{B}^{x})+H(\rho_{B}).
\end{align}
So finally%
\begin{equation}
H(B)_{\rho}-I(X;B)_{\rho}=\sum_{x\in\mathcal{X}}r(x)H(\rho_{B}^{x}),
\end{equation}
as expected.

Now, it is known from \cite{K11}\ that the following limit holds%
\begin{equation}
\lim_{k\rightarrow\infty}H(B|X)_{\rho^{k}}=H(B|X)_{\rho}%
,\label{eq:limit-inf-dim}%
\end{equation}
where%
\begin{equation}
\rho_{XB}^{k}:=\mathcal{P}_{X}^{k}(\rho_{XB}):=\Pi_{X}^{k}\rho_{XB}\Pi_{X}%
^{k}+\frac{\Pi_{X}^{k}%
}{\operatorname{Tr}[\Pi_{X}^{k}]}\otimes 
\operatorname{Tr}_{X}[(I_{X}-\Pi_{X}^{k})\rho_{XB}],\label{eq:project-channel}%
\end{equation}
and $\left\{  \Pi_{X}^{k}\right\}  _{k}$ is a sequence of finite-dimensional
projections strongly converging to the identity. Then by taking the projection
$\Pi_{X}^{k}:=\sum_{x=1}^{k}|x\rangle\langle x|_{X}$, we find from
\eqref{eq:TD-assumption}\ and data processing for normalized trace distance
with respect to the channel defined in \eqref{eq:project-channel} that%
\begin{equation}
\varepsilon\geq\frac{1}{2}\left\Vert \rho_{XB}^{k}-\sigma_{XB}^{k}\right\Vert
_{1},
\end{equation}
where $\sigma_{XB}^{k}:=\mathcal{P}_{X}^{k}(\sigma_{XB})$. Now applying the
uniform continuity bound from Proposition~\ref{prop:main} to the
finite-dimensional states $\rho_{XB}^{k}$ and $\sigma_{XB}^{k}$, we arrive at
the following inequality holding for all $k\in\mathbb{N}$:%
\begin{equation}
\left\vert H(B|X)_{\rho^{k}}-H(B|X)_{\sigma^{k}}\right\vert \leq
\varepsilon\log_{2}(d_{B}-1)+h_{2}(\varepsilon)
\end{equation}
Finally applying the limit in \eqref{eq:limit-inf-dim}, we arrive at the
statement of the corollary.
\end{proof}

\bigskip

Two intriguing questions remain about continuity of conditional entropy. The
first is whether the following inequality could hold%
\begin{equation}
\left\vert H(X|B)_{\rho}-H(X|B)_{\sigma}\right\vert \overset{?}{\leq
}\varepsilon\log_{2}(d_{X}-1)+h_{2}(\varepsilon),
\end{equation}
where $\rho_{XB}$ and $\sigma_{XB}$ are the same classical--quantum states
from \eqref{eq:cq1} (with the systems in the conditional entropy flipped, we
could call these states \textquotedblleft quantum--classical\textquotedblright%
\ now). The other question is whether the following inequality could hold for
fully quantum states $\rho_{AB}$ and $\sigma_{AB}$ that satisfy $\frac{1}%
{2}\left\Vert \rho_{AB}-\sigma_{AB}\right\Vert _{1}\leq\varepsilon$ where
$\varepsilon\in(0,1-1/d_{A}^{2}]$:%
\begin{equation}
\left\vert H(A|B)_{\rho}-H(A|B)_{\sigma}\right\vert \overset{?}{\leq
}\varepsilon\log_{2}(d_{A}^{2}-1)+h_{2}(\varepsilon).
\end{equation}
This inequality is saturated by an example given in Remark~3 of
\cite{Winter15}. These questions were raised during the open problems session
at the workshop \textquotedblleft Algebraic and Statistical ways into Quantum
Resource Theories,\textquotedblright\ held in Banff, Canada during July 2019.
It seems that solving them requires techniques beyond what is currently known.

I acknowledge support from the National Science Foundation under Grant no.~1714215. I am grateful to an anonymous referee for correcting an error and a typo in a previous version of the manuscript.


\begin{thebibliography}{10}

\bibitem{AS19}
Mohammad~A. Alhejji and Graeme Smith.
\newblock A tight uniform continuity bound for equivocation.
\newblock September 2019.
\newblock arXiv:1909.00787v1.

\bibitem{Z07}
Zhengmin Zhang.
\newblock Estimating mutual information via {Kolmogorov} distance.
\newblock {\em IEEE Transactions on Information Theory}, 53(9):3280--3282,
  September 2007.

\bibitem{A07}
Koenraad M.~R. Audenaert.
\newblock A sharp continuity estimate for the von {Neumann} entropy.
\newblock {\em Journal of Physics A: Mathematical and Theoretical},
  40(28):8127, July 2007.
\newblock arXiv:quant-ph/0610146.

\bibitem{LS08}
Debbie Leung and Graeme Smith.
\newblock Continuity of quantum channel capacities.
\newblock {\em Communications in Mathematical Physics}, 292(1):201--215, 2009.

\bibitem{SSWR15}
David Sutter, Volkher~B. Scholz, Andreas Winter, and Renato Renner.
\newblock Approximate degradable quantum channels.
\newblock {\em IEEE Transactions on Information Theory}, 63(12):7832--7844,
  December 2017.
\newblock arXiv:1412.0980.

\bibitem{PhysRevLett.120.160503}
Felix Leditzky, Debbie Leung, and Graeme Smith.
\newblock Quantum and private capacities of low-noise channels.
\newblock {\em Physical Review Letters}, 120(16):160503, April 2018.
\newblock arXiv:1705.04335.

\bibitem{LKDW18}
Felix Leditzky, Eneet Kaur, Nilanjana Datta, and Mark~M. Wilde.
\newblock Approaches for approximate additivity of the {Holevo} information of
  quantum channels.
\newblock {\em Physical Review A}, 97(1):012332, January 2018.
\newblock arXiv:1709.01111.

\bibitem{KW17}
Eneet Kaur and Mark~M. Wilde.
\newblock Amortized entanglement of a quantum channel and approximately
  teleportation-simulable channels.
\newblock {\em Journal of Physics A: Mathematical and Theoretical}, July 2017.
\newblock arXiv:1707.07721.

\bibitem{SWAT18}
Kunal Sharma, Mark~M. Wilde, Sushovit Adhikari, and Masahiro Takeoka.
\newblock Bounding the energy-constrained quantum and private capacities of
  bosonic thermal channels.
\newblock {\em New Journal of Physics}, 20:063025, June 2018.
\newblock arXiv:1708.07257.

\bibitem{KSW19}
Sumeet Khatri, Kunal Sharma, and Mark~M. Wilde.
\newblock Information-theoretic aspects of the generalized amplitude damping
  channel.
\newblock March 2019.
\newblock arXiv:1903.07747.

\bibitem{KGW19}
Eneet Kaur, Saikat Guha, and Mark~M. Wilde.
\newblock Asymptotic security of discrete-modulation protocols for
  continuous-variable quantum key distribution.
\newblock January 2019.
\newblock arXiv:1901.10099.

\bibitem{Fannes73}
Mark Fannes.
\newblock A continuity property of the entropy density for spin lattices.
\newblock {\em Communications~in~Mathematical~Physics}, 31:291, 1973.

\bibitem{AF04}
Robert Alicki and Mark Fannes.
\newblock Continuity of quantum conditional information.
\newblock {\em Journal of Physics A: Mathematical and General}, 37(5):L55--L57,
  February 2004.
\newblock arXiv:quant-ph/0312081.

\bibitem{Winter15}
Andreas Winter.
\newblock Tight uniform continuity bounds for quantum entropies: conditional
  entropy, relative entropy distance and energy constraints.
\newblock {\em Communications in Mathematical Physics}, 347(1):291--313,
  October 2016.
\newblock arXiv:1507.07775.

\bibitem{Shi15}
Maksim~E. Shirokov.
\newblock Tight uniform continuity bounds for the quantum conditional mutual
  information, for the holevo quantity, and for capacities of quantum channels.
\newblock {\em Journal of Mathematical Physics}, 58(10):102202, 2017.

\bibitem{Sh18}
Maksim~E. Shirokov.
\newblock Adaptation of the {Alicki-Fannes-Winter} method for the set of states
  with bounded energy and its use.
\newblock {\em Reports on Mathematical Physics}, 81(1):81--104, February 2018.
\newblock arXiv:1609.07044.

\bibitem{Shirokov_2018}
Maksim~E. Shirokov.
\newblock Uniform continuity bounds for information characteristics of quantum
  channels depending on input dimension and on input energy.
\newblock {\em Journal of Physics A: Mathematical and Theoretical},
  52(1):014001, December 2018.
\newblock arXiv:1610.08870.

\bibitem{Sh19}
Maksim~E. Shirokov.
\newblock Advanced {Alicki-Fannes-Winter} method for energy-constrained quantum
  systems and its use.
\newblock July 2019.
\newblock arXiv:1907.02458.

\bibitem{Petz08}
Denes Petz.
\newblock {\em Quantum Information Theory and Quantum Statistics}.
\newblock Springer Verlag, Berlin Heidelberg, 2008.

\bibitem{BDSW96}
Charles~H. Bennett, David~P. DiVincenzo, John~A. Smolin, and William~K.
  Wootters.
\newblock Mixed-state entanglement and quantum error correction.
\newblock {\em Physical Review A}, 54(5):3824--3851, November 1996.
\newblock arXiv:quant-ph/9604024.

\bibitem{K11}
Anna~A. Kuznetsova.
\newblock Quantum conditional entropy for infinite-dimensional systems.
\newblock {\em Theory of Probability \& Its Applications}, 55(4):709--717,
  November 2011.
\newblock arXiv:1004.4519.

\bibitem{F70}
Harold Falk.
\newblock Inequalities of {J. W. Gibbs}.
\newblock {\em American Journal of Physics}, 38(7):858--869, July 1970.

\bibitem{Lindblad1973}
G{\"o}ran Lindblad.
\newblock Entropy, information and quantum measurements.
\newblock {\em Communications in Mathematical Physics}, 33(4):305--322,
  December 1973.

\end{thebibliography}

\end{document}